\documentclass[journal,onecolumn]{IEEEtran}
\usepackage{amsmath,amsfonts,amssymb,amsthm,hyperref,enumerate,multicol,bm}
\usepackage{calligra,indentfirst,epsfig,lettrine}
\usepackage{caption}
\usepackage{relsize}
\usepackage{tabls}
\usepackage{array}\usepackage{makecell}
\usepackage{longtable}
\calligra

\DeclareFontFamily{U}{mathc}{}
\DeclareFontShape{U}{mathc}{m}{it}%
{<->s*[1.03] mathc10}{}

\DeclareMathAlphabet{\mathscr}{U}{mathc}{m}{it}

\usepackage{comment}
\makeatletter
\newcommand*{\rom}[1]{\expandafter\@slowromancap\romannumeral #1@}
\makeatother
\usepackage{mathtools}
\usepackage{graphicx}
\usepackage{subfigure}
\usepackage{xcolor}
\usepackage{color, soul}
\usepackage{mathrsfs}
\DeclareMathAlphabet{\mathpzc}{OT1}{pzc}{m}{it}

\usepackage{ upgreek }
\newcommand{\etal}{\textit{et al.}}

\newtheorem{definition}{Definition}
\newtheorem{corollary}{Corollary}
\newtheorem{theorem}{Theorem}
\newtheorem{example}{Example}
\newtheorem{lemma}{Lemma}
\newtheorem{remark}{Remark}

\allowdisplaybreaks
\usepackage[table]{xcolor}

\usepackage{multirow}

\begin{document}
\title{{Bounds on the Service Rates of Linear Codes and their Connections to Combinatorial Designs}}

 \author { Monika Yadav, Priyanka~Choudhary
       and~Maheshanand~Bhaintwal\\ 
       
       \thanks{The authors are with the Department of Mathematics, Indian Institute of Technology Roorkee, Roorkee 247667, India (e-mails: \href{mailto:ymonika121@gmail.com}{ymonika121@gmail.com}, \href{mailto:priyanka_c@ma.iitr.ac.in}{priyanka\_c@ma.iitr.ac.in}, 
   \href{mailto:maheshanand@ma.iitr.ac.in}{maheshanand@ma.iitr.ac.in}).}}

\date{}
\maketitle
\begin{abstract}

The service rate region of a distributed storage system describes the set of concurrent demand vectors that can be served without overloading any individual server, and is a fundamental measure of how effectively redundancy is used.  In this paper, we investigate aggregate service rates in distributed data storage systems based on linear codes over arbitrary finite fields. We first derive an upper bound on the aggregate service rate for an arbitrary subset of $t$ data objects in terms of the dual distance and the sizes of their smallest recovery sets. For systematic generator matrices, this bound simplifies to an expression depending only on the code length, the dual distance, and the cardinality of the selected subset.
This result extends existing bounds and provides a subset-wise formulation of the First Dual Distance Bound. 
We then establish conditions under which this upper bound is attained.
We apply our results to several important families of linear codes, including MDS codes, Reed–Muller codes, constacyclic codes, and extended Golay codes, and obtain exact values of the maximum aggregate service rates for the corresponding subsets of data objects.


\end{abstract}


\IEEEpeerreviewmaketitle

\section{Introduction}
The rapid growth in data generation and user demand has created a need for distributed storage systems that provide reliable and efficient access to stored information. Traditionally, data replication has been widely employed to improve reliability and availability in the presence of server failures. 
However, maintaining multiple copies of the same data incurs substantial storage overhead. Erasure coding provides a more storage-efficient alternative by introducing structured redundancy and distributing encoded data across multiple servers \cite{MAktas,Dimakis}. Besides protecting against server failures, such redundancy can also improve data availability when some servers are slow or temporarily unavailable. Consequently, considerable research has focused on improving the storage and repair efficiency of coded storage systems, including reducing redundancy, repair bandwidth, and the number of servers involved in data recovery \cite{Dimakis,Gardner,Huang}.

In practical storage systems, however, data objects may experience significantly different request rates. This non-uniformity in data popularity has motivated storage designs that employ replication, coding, or combinations of the two to accommodate heterogeneous demand patterns \cite{MAktas,Noori}. In addition to storage and repair efficiency, it is therefore important to understand how effectively a given storage scheme can serve incoming requests under the limited processing capacities of its servers. The analysis of data-access performance in coded systems is challenging because a request may be served through different recovery sets, leading to coupled server loads and complex queuing behavior.

The \emph{service rate region} (SRR) provides a useful framework for quantifying this request-serving capability. It consists of all request-rate vectors that can be simultaneously accommodated by the storage system without violating the service-capacity constraints of the individual servers \cite{Anderson, Aktas}. As a result, the SRR is closely related to load balancing and access balancing in distributed storage systems. In particular, for storage schemes operating with the same storage resources, the service rate region provides a means of comparing their ability to accommodate different demand patterns: a larger region corresponds to a larger collection of simultaneously supportable request-rate vectors. Thus, characterizing the service rate region is an important problem in the design and performance analysis of coded distributed storage systems.

Over the past decade, considerable attention has been devoted to studying the SRR of various classes of linear codes. In particular, significant progress has been made in characterizing and bounding the service rate regions associated with various code families, including MDS codes, Simplex codes, Reed–Muller codes, locally recoverable codes, cyclic codes, and Hamming codes \cite{MAktas,LRC, HamPC,PCM,Geo,LyH,LyHS}.

Kazemi {\etal}  \cite{Combit} introduced a graph-theoretic formulation in which recovery sets are represented by edges of a graph. Under this formulation, the allocation of request rates among recovery sets can be interpreted as a fractional matching problem, enabling the derivation of bounds on the service rate region. Further, Kazemi {\etal} \cite{Geo} investigated the problem from a finite-geometric viewpoint and exploited the geometric structure of Reed--Muller codes to describe their recovery sets and analyze the corresponding SRR.
An exact characterization of the service rate region can, however, become computationally difficult, particularly for codes of large length or with complicated recovery set structures. This has motivated the study of bounds on the maximum achievable service rates, which provide useful information about the request-serving capability of a coded storage system without requiring a complete description of its service rate region. In this direction, Ly and Soljanin \cite{Ly} derived bounds on the maximum achievable service rates of individual data symbols by exploiting the symmetry and incidence properties of combinatorial designs.
More recently, Ly and Soljanin \cite{LyH} studied SRR of storage systems based on a broad class of MDS codes. They characterized the axial intercepts of the SRR and determined the smallest standard simplex containing it, showing how the region changes with the number of systematic columns in the chosen generator matrix. 
Despite this progress, several fundamental questions concerning the request-serving capabilities of general linear codes remain open. In particular, for a given subset of data objects, the problem of determining the maximum aggregate request rate
that can be simultaneously supported while keeping the corresponding demand vector within the SRR has not been investigated. Establishing bounds on this quantity and identifying conditions under which such bounds are attained, therefore, constitute natural problems in the study of coded distributed storage systems.

The interplay between linear codes and combinatorial designs has been extensively studied in coding theory \cite{Assmus,Hu,Wang,Xiang,Cao}. In particular, the supports of minimum-weight codewords of the dual code give rise to minimal recovery sets in coded storage systems, establishing a natural connection between minimal recovery sets and combinatorial design theory. Ding and Tang \cite{Dingdesigns} provide a comprehensive treatment of designs arising from linear codes, including those supported by several important families of linear codes. In the context of near-MDS (NMDS) codes, Ding and Tang \cite{Ding} constructed infinite families of NMDS codes supporting infinite families of $2$-designs and $3$-designs. In particular, they presented a family of NMDS codes over $\mathbb{F}_{3^s}$, with $s\geq 2$, for which the supports of the minimum-weight codewords of both the codes and their duals form $3$-designs. Subsequently, Heng and Wang \cite{Heng} constructed several new infinite families of NMDS codes supporting $2$-designs and $3$-designs with parameters different from those previously known. These connections between linear codes and combinatorial designs are particularly relevant to the present work, as they enable us to exploit the design structure of minimum-weight dual codewords to study the attainability of bounds on aggregate service rates.

In the context of these developments, we investigate the maximum aggregate request rate that can be simultaneously supported for a given subset of data objects while ensuring that the corresponding demand vector remains within the SRR. 
The  paper is organized as follows. In Section \ref{Notation}, we introduce the necessary preliminaries and notation concerning the service rate region of coded distributed storage systems and $t$-designs. In Section \ref{secBound}, we first derive an upper bound on the aggregate service rate for an arbitrary subset of $t$ data objects in terms of the minimum distance $d^\perp$ of the dual code and the sizes of the smallest recovery sets associated with these data objects (Theorem \ref{thm1}). 
In the systematic case, this bound simplifies to a form depending only on the code length, the dual distance, and the cardinality of the selected subset (Theorem \ref{thm:Upperbound}).   We further establish sufficient conditions for the attainability of this upper bound in Theorem \ref{thm:design}. More precisely, let $\mathscr{C}$ be an $[n,k,d]$ linear code over $\mathbb{F}_q$ with a systematic generator matrix, and let $d^\perp$ denote the minimum distance of $\mathscr{C}^\perp$. For a subset $I\subseteq[k]$ of cardinality $t$ satisfying $t\leq\min\{d^\perp-1,n-d^\perp+1\},$ we show that if the supports of the minimum-weight codewords of $\mathscr{C}^\perp$, counted without repetition, form a $(t+1)$-design, then the upper bound on the aggregate service rate of the data objects indexed by $I$ is attained by an achievable demand vector. In Section \ref{app}, 
we apply the above achievability result to several important families of linear codes and obtain exact values of the maximum aggregate service rates for the corresponding subsets of data objects. For systematic MDS codes, the supports of the minimum-weight codewords of the dual code consist of all $(k+1)$-subsets and hence form a complete design. Consequently, our result applies to every subset $I\subseteq[k]$ of cardinality $t\leq k$ satisfying the required length condition and yields the exact aggregate service rate $t+(n-t)/k.$ This value itself is not new and also follows from the results of Aktaş {\etal} \cite{Aktas}; in particular, for $n\geq 2k,$ their water-filling algorithm attains this value. In the present setting, however, the condition can be relaxed to $n\geq k+t$ for the face of the SRR corresponding to a fixed subset $I$ of cardinality $t,$ thereby covering a range of code lengths for which the complete SRR is not known. For Reed–Muller codes, the relevant minimum-weight dual codewords are supported on affine flats in $\operatorname{AG}(m,2),$ whose design structure allows us to determine exact aggregate service rates for pairs of data objects. Similarly, for the extended ternary and binary Golay codes, the Steiner systems $S(5,6,12)$ and $S(5,8,24),$ respectively, allow the achievability result to be applied to subsets of up to four data objects. In each of these cases, the constructed allocation fully utilizes every server, confirming that the corresponding upper bound is attained.


\section{Preliminaries and Notation}\label{Notation}
In this section, we recall some fundamental concepts related to distributed storage systems and their associated service rate regions, together with relevant results from combinatorial design theory that will be used to establish our main results.

Let $\mathbb F_q$ denote the finite field of cardinality $q.$ A linear code $\mathscr{C}$ of length $n$ over $\mathbb{F}_q$ is an $\mathbb{F}_q$-linear subspace of $\mathbb{F}_q^n,$ where $n$ is a positive integer and $\mathbb{F}_q^n$ denotes the $n$-dimensional vector space of all $n$-tuples over $\mathbb{F}_q.$ The dimension $k$ of $\mathscr{C}$ is defined as its dimension as an $\mathbb{F}_q$-linear subspace of $\mathbb{F}_q^n.$ Elements of the code $\mathscr{C}$ are called codewords.   A generator matrix $G$ of $\mathscr{C}$  is a $k\times n$ matrix over $\mathbb{F}_q$ whose rows form a basis for $\mathscr{C}.$ Every codeword of $\mathscr{C}$ can be uniquely represented as $\mathbf{v}G$ for some $\mathbf{v}\in\mathbb{F}_q^k.$ Hence,  $\mathscr{C}=\{\mathbf{v}G~|~ \mathbf{v}\in\mathbb{F}_q^k\},
$ and $G$ completely determines the code $\mathscr{C}.$
The Hamming weight of a codeword $\mathbf{d}\in \mathscr{C},$ denoted by $w_H(\mathbf{d}),$ is the number of its nonzero components. The minimum Hamming distance of $\mathscr{C},$ denoted by $d(\mathscr{C}),$ is the minimum Hamming weight among all nonzero codewords of $\mathscr{C}.$ Throughout this paper, for a positive integer $m$, we use the notation $[m]$ to denote the set $\{1,2,\ldots,m\}.$ For any codeword $\mathbf{d}=(d_1,d_2,\ldots,d_n)\in \mathscr{C},$ its support is the set of indices corresponding to its nonzero coordinates, namely, $Supp(\mathbf{d})=\{i\in[n]~|~d_i\neq 0\}.
$
A linear code $\mathscr{C}$  over $\mathbb{F}_q$ having length $n,$ dimension $k$ and minimum Hamming distance $d$ will be called an $[n,k,d]$ linear code. Moreover, for each $\ell \in \{0,1,\ldots,n\},$ we denote by $W_\ell$ the number of codewords in $\mathscr{C}$ having Hamming weight $\ell.$
Associated with an $[n,k,d]$ linear code $\mathscr{C}$ is its dual code $\mathscr{C}^{\perp},$ consisting of all vectors in $\mathbb{F}_q^n$ that are orthogonal to every codeword of $\mathscr{C}$ with respect to the Euclidean inner product. More precisely, $$
\mathscr{C}^{\perp}
=\big\{\mathbf{b}\in\mathbb{F}_q^n~|~
\mathbf{b}\cdot\mathbf{d}=0\ \text{for every }\mathbf{d}\in\mathscr{C}\big\}.$$
The dual $\mathscr{C}^{\perp}$ is a linear code of length $n$ and dimension $n-k$ over $\mathbb{F}_q.$ We denote the minimum Hamming distance of $\mathscr{C}^{\perp},$ also referred to as the \emph{dual distance} of $\mathscr{C}$, by $d^{\perp}=d(\mathscr{C}^{\perp}).$ 

We now recall some basic concepts related to coded distributed storage systems and their associated SRRs. Let $G=[G^{(1)}\ G^{(2)}\ \cdots\ G^{(n)}]$ be a full-row rank $k\times n$  matrix over $\mathbb{F}_q$ with no zero columns, where $G^{(j)}$ denotes the $j$-th column of $G.$   Let $\mathscr{C}$ be the $[n,k,d]$ linear code over $\mathbb{F}_q$ generated by $G.$ Further,  a generator matrix $G$ is said to be in \emph{systematic form} if the $k\times k$ submatrix formed by its first $k$ columns is $I_k.$ 
We use the matrix $G$ to define a coded distributed storage system consisting of $n$ servers (or nodes), indexed by $[n],$ for storing $k\geq 2$ data objects. Throughout, we assume $k\le n,$ with uniform data-object sizes and equal storage capacities across all servers. Accordingly, the $k$ data objects are represented by a vector
$\mathbf{u}=(u_1,u_2,\ldots,u_k)\in\mathbb{F}_q^k.$
 The data are  encoded using $G,$ resulting in the vector
$\mathbf{u}G\in\mathbb{F}_q^n.$
For each $j\in[n],$ the $j$-th coordinate of $\mathbf{u} G,$ given by
$\mathbf{u}G^{(j)},$
is stored on server $j.$ Hence, the column $G^{(j)}$ specifies the linear combination of the original data objects stored on the $j$-th server. A server $j$ is called  \emph{systematic} (or  \emph{uncoded}) if $G^{(j)}$ is a standard basis vector of $\mathbb F_q^k;$ otherwise, it is called a coded server.

For each $i\in[k],$ let $\mathbf{e}_i$ denote the $i$-th standard basis vector of $\mathbb{F}_q^k.$ A subset $R\subseteq [n]$ is called a \emph{recovery set} for the data object $u_i$ if $u_i$ can be recovered from the contents of the servers indexed by $R.$ Equivalently, the standard basis vector $\mathbf{e}_i$ must lie in the $\mathbb F_q$-linear span of the columns of $G$ corresponding to the indices in $R.$ 
Further,
a recovery set $R$ for $u_i$ is said to be \emph{minimal} if no proper subset of $R$ is itself a recovery set for $u_i.$
Hence, a minimal recovery set contains no redundant server for the recovery of $u_i.$

For each $i\in[k],$ let
$\mathcal{R}_i(G)=\{R_{i,1},R_{i,2},\ldots,R_{i,t_i}\}$
denote the collection of all minimal recovery sets for the data object $u_i,$ where $t_i=|\mathcal{R}_i(G)|.$ The $k$-tuple $\mathcal{R}(G)=\big(\mathcal{R}_1(G),\mathcal{R}_2(G),\ldots,\mathcal{R}_k(G)\big)$ comprising the collections of minimal recovery sets for all $k$ data objects is referred to as the \emph{recovery $G$-system} corresponding to the minimal recovery sets for all $k$ data objects.
Each server has a limited rate at which it can process incoming requests. For $j\in[n]$, let $\delta_j\in\mathbb{R}_{\geq 0}$ denote the service capacity of server $j.$ The service capacities of all $n$ servers are represented by the vector
$\boldsymbol{\delta}=(\delta_1,\delta_2,\ldots,\delta_n)\in\mathbb{R}_{\geq0}^n.$
Throughout this paper, we assume that all servers have unit service capacity, i.e.,
$\delta_j=1,$ for $ j\in[n].$
Finally, let $\lambda_i\in\mathbb{R}_{\geq 0}$ denote the rate at which requests for the data object $u_i$ arrive at the storage system. The request rates for all $k$ data objects are collected in the demand vector
$\boldsymbol{\lambda}
=(\lambda_1,\lambda_2,\ldots,\lambda_k)
\in\mathbb{R}_{\geq0}^k.$

\begin{definition}\cite{alfarano2024service}
  Let $\mathcal{R}(\mathcal{G})=\big(\mathcal{R}_1(G),\mathcal{R}_2(G),\ldots,\mathcal{R}_k(G)\big)$ be the recovery $G$-system, corresponding to the minimal recovery sets for all $k$ data objects. For each $i\in[k]$ and $j\in[t_i],$ let $\lambda_{i,j}\geq 0$ denote the portion of the request rate for the data object $u_i$ that is served through the recovery set $R_{i,j}.$ A demand vector
$\boldsymbol{\lambda}=(\lambda_1,\lambda_2,\ldots,\lambda_k)\in\mathbb{R}_{\geq 0}^k$ 
is said to be \emph{achievable} if the incoming requests can be distributed among the corresponding recovery sets without exceeding the service capacity of any server. The \emph{service rate region} (SRR) associated with $G,$ denoted by $\Lambda\big(\mathcal{R}(G)\big),$ is the collection of all achievable demand vectors. Equivalently, $\boldsymbol{\lambda}\in\Lambda\big(\mathcal{R}(G)\big)$  if there exist nonnegative quantities $\lambda_{i,j}$ for $i\in[k]$ and $j\in[t_i],$ satisfying the following constraints:
  
  \begin{equation}\label{eq1}
    \sum\limits_{j=1}^{t_i}\lambda_{i,j}=\lambda_i
\text{~ for all }~ i\in [k],
\end{equation}
and \begin{equation}\label{eq2}
    \sum_{i=1}^{k}\sum\limits_{\substack{j=1\\ \ v\in R_{i,j}}}^{t_i}
\lambda_{i,j}
\le 1
\text{~ for all }~ v\in [n].
\end{equation}

 \end{definition} 
  
 Condition \eqref{eq1}  guarantees that the entire request rate for each data object is distributed among its corresponding recovery sets, whereas condition \eqref{eq2} ensures that no server is assigned a load exceeding its service capacity.

For each $I \subseteq [k],$ let us define 
$$\left(\sum_{i \in I}\lambda_{i}\right)^{*}=\max\left\{ \sum_{i \in I}\lambda_{i} ~\big |~ \boldsymbol{\lambda}= (\lambda_1,\lambda_2,\ldots,\lambda_k)\in\Lambda(\mathcal{R}(G))\right\}.$$

In other words, $\left(\sum_{i\in I}\lambda_i\right)^*$ characterizes the maximum combined request rate that the storage system can support for the data objects indexed by $I,$ under the feasibility constraints imposed by the SRR.
We refer to this quantity as the \emph{maximum achievable aggregate rate} for the set of data objects indexed by $I.$

We now recall some basic definitions and properties of $t$-designs that will be used to establish a connection between SRRs and combinatorial designs.
\begin{definition}\cite[Sec. 8.1]{Huffman}
 Let $n,$  $\omega$, and $t$ be positive integers with $1\leq \omega<n.$ A $t$-$(n,\omega,\mu_t)$ design is an incidence structure $\mathcal{X}=(U,\mathscr{B}),$ where $U$ is a set of $n$ points and $\mathscr{B}$ is a collection of $\omega$-element subsets of $U,$ referred to as blocks, such that every $t$-element subset of $U$ occurs in exactly $\mu_t$ blocks of $\mathscr{B}.$

\end{definition}

\begin{lemma}\label{lem2.1}\cite[Th. 8.1.3]{Huffman}
Suppose that $\mathcal{X}=(U,\mathscr{B})$ is a $t$-$(n,\omega,\mu_t)$ design. For any $s\in[t],$ $\mathcal{X}$ is also a $s$-$(n,\omega,\mu_s)$ design, where $\mu_s=\mu_t\cdot \frac{\binom{n-s}{t-s}}{\binom{\omega-s}{t-s}} .$
\end{lemma}
\begin{lemma}\label{twoset}
    \cite[Theorem 9.7]{Stinson} Let $\mathcal{X}=(U, \mathscr{B})$ be a $t$-$(n,\omega,\mu_t)$ design. Let \(S,T\subseteq U\) be disjoint sets with $|S|=a,$ $ |T|=b$ and  $a+b\le t$. Then the number of blocks in \(\mathscr{B}\) that contain every point of $S$ and contain no point of $T$ is

$$
\mathcal{M}_{a}^b
=
\mu_t
\frac{\binom{n-a-b}{\omega-a}}
     {\binom{n-t}{\omega-t}}.$$

\end{lemma}
Throughout the paper, we use the notation introduced in this section.

\section{Service Rate Bounds for subsets of data objects}\label{secBound}
In this section, we investigate bounds on the aggregate service rate associated with a subset of data objects in a coded storage system. We first derive general upper bounds using the sizes of the smallest recovery sets and the minimum Hamming distance of the dual code. We then specialize these bounds to systematic generator matrices and identify conditions under which they are attained.

Let $i\in[k].$ A recovery set $\mathcal{O}$ for the data object $u_i$ is called a smallest recovery set if
$|\mathcal O|=\min\big\{|R|~|~ R\in\mathcal{R}_i(G)\big\}.$
That is, $\mathcal O$ has minimum cardinality among all the minimal recovery sets for the data object $u_i.$
The following lemma relates the size of any minimal recovery set of a data object to that of a smallest recovery set of that data object and the minimum distance of the dual code.  A similar result was established in \cite{Ly} for the binary case. For completeness, we present its extension to codes over arbitrary finite fields.
\begin{lemma}\label{lem:minRsize}
Let $\mathscr{C}$ be an $[n,k,d]$ linear code over $\mathbb{F}_q$ with a generator matrix $G.$ Consider a storage system encoded using the generator matrix $G.$ For $i\in [k],$ 
let $\mathcal{O}$ be a smallest recovery set for the data object $u_i$ with  $|\mathcal{O}|=a$. 
Then any minimal recovery set $R$ for $u_i,$  distinct from $\mathcal{O},$ satisfies
\begin{equation*}\label{eq:sharp-bound}
  |R|\;\ge\;d^{\perp}-a ,
\end{equation*}
and consequently
\begin{equation*}\label{eq:max-bound}
  |R|\;\ge\;\max\bigl\{\,a,\;d^{\perp}-a\,\bigr\}.
\end{equation*}
\end{lemma}
\begin{proof}  Since $\mathcal{O}$ is a minimal recovery set for the data object $u_i,$ we have $\sum\limits_{\ell\in \mathcal{O}}\alpha_{\ell}G^{(\ell)}=\mathbf{e}_i$ for some $\alpha_{\ell} \in \mathbb{F}_q^*$. 

Now, let $R\neq\mathcal{O}$ be any minimal recovery set for $u_i.$ Then, there exist coefficients $\beta_j\in\mathbb{F}_q^*$ for $j\in R,$ satisfying
$\sum\limits_{j\in R}\beta_jG^{(j)}=\mathbf{e}_i.$
Subtracting these two expressions gives
$$\sum\limits_{\ell\in\mathcal{O}}\alpha_\ell G^{(\ell)}
-\sum\limits_{j\in R}\beta_jG^{(j)}=0,$$
which can be rewritten as
$$\sum\limits_{\ell\in \mathcal{O}\setminus \mathcal{O}\cap R}\alpha_{\ell}G^{(\ell)} -\sum\limits_{j\in R\setminus \mathcal{O}\cap R}\beta_{j}G^{(j)} +\sum\limits_{b\in \mathcal{O}\cap R}(\alpha_b-\beta_b)G^{(b)}=0.$$
Thus, the above linear dependence among the columns of $G$ determines a codeword
$\mathbf{h}=(h_1,h_2,\ldots,h_n)\in\mathscr{C}^\perp,
$
where
$$h_j=\left\{ \begin{array}{cl}
    \alpha_j & \text{if } j \in \mathcal{O}\setminus \mathcal{O}\cap R;  \\
   -\beta_j  & \text{if } j \in R\setminus \mathcal{O} \cap R;\\
   \alpha_j-\beta_j & \text{if } j \in \mathcal{O}\cap R; \\0 & \text{otherwise.}
 \end{array}\right.$$
Since $R\neq\mathcal{O},$ and both $R$ and $\mathcal{O}$ are minimal recovery sets for $u_i,$  neither can be properly contained in the other. Hence, $\mathbf{h} \neq\mathbf{0}.$ By the minimality of $d^\perp,$
we have $w_H(\mathbf{h})\ge d^\perp.$
Furthermore,
$Supp(\mathbf{h})\subseteq\mathcal{O}\cup R,$
and therefore
$$d^\perp
\le w_H(\mathbf{h})
\le |\mathcal{O}\cup R|
=|\mathcal{O}|+|R|-|\mathcal{O}\cap R|.$$
Since $|\mathcal{O}|=a,$ it follows that
 $|R|\;\ge\;d^{\perp}-a .$
Further, $\mathcal{O}$ is a smallest recovery set for $u_i,$ and hence $|R|\ge |\mathcal{O}|=a.$
Combining the above two inequalities, we obtain
$|R|\ge \max\{a,d^\perp-a\},$
which proves the desired result. 
\end{proof}

\begin{remark}
If $|R|<d^{\perp},$ it is easy to see that   there exist unique coefficients $\beta_j\in\mathbb{F}_q^* $    for $j\in R,$ satisfying
$\sum\limits_{j\in R}\beta_jG^{(j)}=\mathbf{e}_i.$ 
\end{remark}

\begin{remark}
If $\mathcal O=\{i\}$ is the smallest recovery set for $u_i,$ then  $a=1$ and  every minimal recovery set $R\neq\{i\}$ satisfies $|R|\ge d^\perp-1.$
Moreover, this bound is attained whenever there exists a minimum Hamming weight codeword $\mathbf{h}\in\mathscr{C}^\perp$ whose support contains $i,$ since $Supp(\mathbf{h})\setminus\{i\}$ forms a recovery set for $u_i$ of size $d^\perp-1.$
Further, if $d^\perp\le 2a,$ then we have  $\max\{a,d^\perp-a\}=a.$
Thus, in this case, the bound reduces to $|R|\ge a,$ and the dual distance does not provide any improvement over the minimum-size condition imposed by $\mathcal{O}.$
\end{remark} 

We next extend the preceding recovery set analysis to   data objects belonging to a subset $I\subseteq[k] .$ In the following theorem, we consider smallest recovery sets for these data objects and derive an explicit bound on the maximum achievable sum of their service rates that can be supported by the storage system.

\begin{theorem}\label{thm1}
Let $\mathscr{C}$ be an $[n,k,d]$ linear code over $\mathbb{F}_q$ with a generator matrix $G.$ Consider a storage system based on encoding defined by $G.$  Let $ I\subseteq [k],$ and for $i\in I$, let
$\mathcal{O}_i$ be a smallest recovery set for data object $u_{i}.$  Let 
$|\mathcal{O}_i|=a_i$, and  define $b_i=\max\{a_i,d^{\perp}-a_i\}$, where $d^\perp$ denotes the minimum Hamming distance of $\mathscr{C}^\perp.$ Then the sum of achievable rates for the data objects indexed by $I$  satisfies
\begin{equation}\label{Ibound}
    \sum_{i\in I}\lambda_ib_i\leq n+\sum_{i\in I}(b_i-a_i).\end{equation}
\end{theorem}
\begin{proof}  For each $i\in I,$ decompose the service rate as
$\lambda_i=x_i+y_i,
$
where $x_i\leq 1$ denotes the rate assigned to the recovery set $\mathcal O_i,$ and $y_i=\lambda_i-x_i$ is the total rate assigned to all other recovery sets for $u_i.$
By Lemma \ref{lem:minRsize}, every minimal recovery set for $u_i$ distinct from $\mathcal{O}_i$ has cardinality at least $b_i$. Since the total service load cannot exceed the aggregate capacity $n$ of the servers, we obtain
$$ \sum_{i\in I}x_ia_i+\sum_{i\in I}(\lambda_i-x_i)b_i\leq n,$$ 
which implies
$$ \sum_{i\in I}\lambda_ib_i\leq n+\sum_{i\in I}x_i(b_i-a_i).$$   
 Since $x_i\leq 1,$ we get 
$$ \sum_{i\in I}\lambda_ib_i\leq n+\sum_{i\in I}(b_i-a_i).$$

\end{proof}
\begin{remark}
    When $d^\perp \le 2a_i$ for all $i\in I$, we have $b_i=a_i$ for all $i$. The bound collapses to  $\sum_{i\in I}a_i\lambda_i\le n$. Tightness then requires only that the recovery sets of minimum size can be weighted so that all $n$ servers are exactly saturated.
\end{remark}
\begin{example}
  Consider the $r$th-order binary Reed–Muller code $\mathrm{RM}(r,m)$ of length $2^m,$ where $m\geq r+1,$ represented with respect to the monomial basis (see \cite[Chap. 13]{Sloane}). For a data object $u_i$ of order $r$ (as defined by Equation (6) of \cite{LyHS}), the cardinality of a smallest recovery set is $a_i=2^r$ \cite[Thm. 4]{LyHS}. Moreover, since $d^\perp=2^{r+1},$ we have
$b_i=\max\{a_i,d^\perp-a_i\}
=\max\{2^r,2^{r+1}-2^r\}
=2^r.$
Let $I_r\subseteq [m+1]$ denote the index set corresponding to the data objects of order $r.$ Applying \eqref{Ibound} to $I_r,$ we obtain
$\sum\limits_{i\in I_r}2^r\lambda_i\leq 2^m,$
and hence
$\sum\limits_{i\in I_r}\lambda_i\leq 2^{m-r}.$
On the other hand, Theorem 7 of \cite{LyHS} shows that the maximum aggregate service rate for the data objects indexed by $I_r$ is exactly $2^{m-r}.$ Therefore, the upper bound obtained from \eqref{Ibound} is tight for this class of data objects in $\mathrm{RM}(r,m).$
\end{example}
As an immediate consequence of the preceding theorem, the following corollary establishes an explicit upper bound on the sum of service rates achievable by two distinct data objects with disjoint smallest recovery sets.
\begin{corollary}\label{corollary1}
Under the assumptions of Theorem \ref{thm1}, let $I=\{i,j\}\subseteq[k],$ where $i\neq j$ and  $\mathcal{O}_i\cap\mathcal{O}_j=\emptyset.$  
Then the sum of achievable rates
$\lambda_i$ and $\lambda_j$ for data objects $u_i$ and $u_j$, respectively, satisfies
\begin{equation}\label{2bound}
    (\lambda_i+\lambda_j)^*\leq\left\{ \begin{array}{cl}
   2+\frac{b_j-b_i}{b_i}+\frac{n-(a_i+a_j)}{b_i}  & \text{if } b_i\leq b_j;\vspace{1mm}\\
    2+\frac{b_i-b_j}{b_j}+\frac{n-(a_i+a_j)}{b_j}  & \text{if } b_i> b_j.\\
 \end{array}\right.\end{equation}
\end{corollary}

\begin{example}
    Consider an $[n,k=3,n-2]$ MDS code $\mathscr{C}$ over $\mathbb{F}_q,$ where $n=q+2$ when $q$ is even and $n=q+1$ when $q$ is odd. 
    Giusto {\etal} \cite{oval} showed that
    there exists a non-systematic generator matrix for $\mathscr{C}$ such that all three data objects have the smallest recovery sets of size $2$, and for any two distinct  data objects $u_i$ and $u_j$, there exist recovery sets of size $2$ that are disjoint.    Then, by the  upper bound  \eqref{2bound}, we get $(\lambda_i+\lambda_j)^*\leq \frac{n}{2}$ for all $i\ne j\in[3],$ which agrees with Theorem 13 of  \cite{oval}. 
\end{example}

We next consider the case in which the smallest recovery sets for two distinct data objects are not disjoint. In the following corollary, we derive an explicit bound on the maximum achievable sum of the service rates of two data objects when these smallest recovery sets have a nonempty intersection.

\begin{corollary}\label{thm2}
Under the assumptions of Theorem \ref{thm1}, let $I=\{i,j\}\subseteq[k],$ where $i\neq j$ and  $\mathcal{O}_i\cap\mathcal{O}_j\neq\emptyset.$  
Then the sum of achievable rates
$\lambda_i$ and $\lambda_j$ for objects $u_i$ and $u_j$, respectively, satisfies
$$(\lambda_i+\lambda_j)^*\leq\left\{ \begin{array}{cl}
   1+\frac{n-a_i}{b_i}  & \text{if } b_i\leq b_j;\vspace{1mm}\\
    1+\frac{n-a_j}{b_j}  & \text{if } b_i> b_j.\\
 \end{array}\right.$$
\end{corollary}


We next specialize Theorem \ref{thm1} to the case of a systematic generator matrix. In this setting, the singleton $\{i\}$ is a smallest recovery set for each data object $u_i,$ so that $a_i=1$ for every $i\in I.$ Consequently, when $d^\perp>1,$ Theorem \ref{thm1} directly yields the following upper bound on the maximum achievable aggregate service rate for any subset $I\subseteq[k].$

\begin{theorem}\label{thm:Upperbound}
 Let $\mathscr{C}$ be an $[n,k,d]$ linear code over $\mathbb{F}_q$ with a systematic generator matrix $G.$ Consider a storage system based on encoding defined by $G.$   Then, for any subset $I\subseteq[k],$ provided that $d^\perp>1$, we have  $$\big(\sum_{i\in I}\lambda_i\big)^*\leq |I|+\frac{n-|I|}{d^\perp-1}.$$
\end{theorem}
\begin{proof}
 Since $G$ is systematic, for each $i\in I,$ the singleton set $\{i\}$ is a smallest recovery set for $u_i.$ Moreover, by Lemma \ref{eq:max-bound}, every other minimal recovery set for $u_i$ has cardinality at least $d^\perp-1.$
To obtain an upper bound on the aggregate service rate, assign at most one unit of service to the systematic recovery set $\{i\}$  for each $i\in I.$ The remaining rate, $\lambda_i-1,$ must then be served through recovery sets of size at least $d^\perp-1.$ Since the systematic recovery sets $\{i\},$ $i\in I,$ are pairwise disjoint, serving one unit through each of them consumes a total server capacity of $|I|.$ Therefore, the total server load satisfies 
$$|I|+(d^\perp-1)\sum_{i\in I}(\lambda_i-1)\leq n,$$ which implies that $$\sum_{i\in I}\lambda_i\leq |I|+\frac{n-|I|}{d^\perp-1}.$$
\end{proof}
\begin{remark} The upper bound assertion of Theorem \ref{thm:Upperbound} follows from First Dual Distance Bound  (Theorem IV.6 of Alfarano {\etal}
    \cite{Alfarano}).  Indeed, restricting the demand vector to the coordinates indexed by $I,$ \textit{i.e.}, setting $\lambda_j=0$ for all 
    $j\in[k]\setminus I$ and $\min\{\lambda_i,1\}=1$ for all $i\in I,$ reduces their upper bound on the aggregate service rate to the upper bound established in Theorem \ref{thm:Upperbound}.

\end{remark}
 \begin{remark}
The bound established in Theorem~\ref{thm:Upperbound} can also be read backward, as a necessary condition on the dual distance. Let $I\subseteq[k]$ and let $S>|I|$
be a target aggregate service rate for the objects indexed by $I$. If a code with
a systematic generator matrix and $d^{\perp}\ge2$ supports $S$, then
\[
  S\;\le\;|I|+\frac{n-|I|}{d^{\perp}-1}
  \qquad\Longleftrightarrow\qquad
  d^{\perp}\;\le\;1+\frac{n-|I|}{S-|I|},
\]
equivalently $n\ge|I|+(S-|I|)(d^{\perp}-1)$. The direction of this inequality is
worth noting. A large dual distance forces every non-trivial recovery set to be
large, and large recovery sets consume more capacity per unit of service, so the
achievable aggregate rate decreases as $d^{\perp}$ grows; the two design
objectives of error correction and of service therefore pull against each other.
In the extreme case $d^{\perp}=2$, the code has repeated columns, every recovery
set has size one, and the constraint degenerates to $S\le n$. Consequently, among
codes of length $n$ supporting the target $S$ on $I$, the dual distance cannot
exceed $1+\frac{n-|I|}{S-|I|}$; combined with the Singleton or Griesmer bound this
restricts the admissible triples $(n,k,S)$.

\end{remark}

We next investigate conditions under which the upper bound established in the preceding theorem is attained. In the following theorem, we show that this bound is achievable when the supports of the minimum weight codewords of the dual code form a suitable combinatorial design.

\begin{theorem}\label{thm:design}
Let $\mathscr{C}$ be an $[n,k,d]$ linear code over $\mathbb{F}_q$ with a systematic generator matrix $G,$ and  let $d^{\perp}\ge 2.$  Consider a storage system based on encoding defined by $G.$  Let $I\subseteq[k]$ be a subset of cardinality $t,$ where $t\leq \min \{d^\perp-1, n-d^\perp+1\}.$ 
Suppose that the supports of the minimum weight codewords of $\mathscr{C}^{\perp}$, counted without repetition, form a $(t+1)$-$(n,d^\perp,\mu_{t+1})$ design $\mathcal{X}=([n],\mathscr{B}).$  
Then there exist a demand vector $\boldsymbol\lambda=(\lambda_1,\dots,\lambda_k)\in \Lambda(\mathcal{R}(G))$  satisfies
\begin{equation}\label{eq:design-sum}
 \sum_{i\in I}\lambda_{i}
  ~=~t+\frac{n-t}{d^{\perp}-1}.
\end{equation}
\end{theorem}
\begin{proof}
 By Theorem 8.1.6 of \cite{Huffman},  it follows that each server from $[n]$  occurs in exactly $\frac{\mu_{t+1}(n-1)}{d^{\perp}-1}$ blocks.  Now fix $i \in I,$ and let $B\in \mathscr{B}$ such that $i\in B.$ It is easy to see that there exists a codeword $\mathbf{h}=(h_1,h_2,\ldots, h_n)\in \mathscr{C}^{\perp}$ such that $Supp(\mathbf{h})=B$ and $h_i\neq0.$ This implies that $\sum\limits_{j\in B}h_j G^{(j)}=0,$ which implies that  $
  G^{(i)}=\mathbf e_i= -\sum_{j\in B\setminus\{i\}}\bigl(h_{i}^{-1}h_j\bigr)G^{(j)}.$ 
Hence,  $B\setminus\{i\}$ is a
recovery set of size $d^{\perp}-1$ for the data object $u_{i}$.

To prove the result, we first count the number of blocks containing the node $i$ and none of the nodes from $I\setminus \{i\}$. Take $S_1=\{i\}$ and $T_1=I\setminus \{i\}.$ Note that $S_1\cap T_1=\emptyset.$ Then by Lemma \ref{twoset},  we see that the number of blocks in $\mathscr{B}$ that contain $S_1$ and contain no point of $T_1$ is 
$$ \mathcal{M}_{1}^{t-1}
=\mu_{t+1}
\frac{\binom{n-t}{d^{\perp}-1}}
     {\binom{n-t-1}{d^{\perp}-t-1}}.$$
Next, fix a node $v\in[n]\setminus I.$ Now,  we count the number of blocks containing the nodes $i$ and $v,$  and none of the nodes from $I\setminus \{i\}.$ Take $S_2=\{i,v\}$ and $T_2=I\setminus \{i\}.$ Note that $S_2\cap T_2=\emptyset.$ Then again,  by Lemma \ref{twoset},  we see that the number of blocks in $\mathscr{B}$ that contain $S_2$ and contain no point of $T_2$ is $$ \mathcal{M}_{2}^{t-1}
=\mu_{t+1}
\frac{\binom{n-t-1}{d^{\perp}-2}}
     {\binom{n-t-1}{d^{\perp}-t-1}}.$$ 

We now construct a demand vector $\boldsymbol\lambda=(\lambda_1,\dots,\lambda_k)\in \Lambda(\mathcal{R}(G))$ that   satisfies the equation \eqref{eq:design-sum}. 

For each $\ell\in I,$ we assign the rate $1$ to the systematic recovery set $\{\ell\}$,
and assign the rate $\dfrac{1}{t\mathcal{M}_2^{t-1}}$ to each of the $\mathcal{M}_1^{t-1}$ recovery sets of $u_{\ell}$ of the form 
$B\setminus\{\ell\},$ where $B\in\mathscr{B}$ and $B\cap I=\{\ell\}$. All other
recovery sets of $u_{\ell}$ receive a rate $0$. Recall that for each $v\in [n]\setminus I,$ the number of blocks that contain $\ell$ and $v$ 
 (\textit{i.e.}, recovery sets for $u_{\ell}$ that contain  the node $v$) and none of the elements of $I\setminus\{\ell\}$ is $\mathcal{M}_2^{t-1}.$ Summing over all data objects belonging to $I,$  the total load
on the  $v$-th node  is $(\mathcal{M}_2^{t-1}t)\cdot\frac{1}{\mathcal{M}_2^{t-1}t}=1$. Thus,
every server carries a load exactly $1$, the allocation is feasible, and we get 

\begin{equation*}
 \sum_{i\in I}\lambda_{i}
  = t+\frac{t\mathcal{M}_1^{t-1}}{t\mathcal{M}_2^{t-1}}= t+\frac{n-t}{d^{\perp}-1}.
\end{equation*}
\end{proof}

\begin{remark}
 If the supports of the minimum weight codewords of $\mathscr{C}^{\perp}$ form an $\alpha$-$(n,d^{\perp},\mu_{\alpha})$ design and $I\subseteq[k]$ satisfies $|I|=t<\alpha.$  It follows from Lemma \ref{lem2.1} that these supports also form $(t+1)$-$(n,d^{\perp},\mu_{t+1})$ design, where $\mu_{t+1}=
\mu_\alpha\cdot \frac{\binom{n-t-1}{\alpha-t-1}}{\binom{d^{\perp}-t-1}{\alpha-t-1}} .$ Consequently, Theorem \ref{thm:design} can be applied to determine the exact value of the maximum aggregate service rate for the data objects indexed by $I.$
\end{remark}

\begin{remark} Theorem \ref{thm:design}  generalizes Corollary 1 of Ly and Soljanin \cite{Ly}. Their result considers a single data object and shows that the corresponding service rate bound is attained when the supports of the minimum weight codewords of the dual code form a $2$-design. Theorem 4 extends this result to an arbitrary subset $I\subseteq [k]$ of cardinality $t,$ showing that the upper bound on the aggregate service rate of the data objects indexed by $I$ is attained when the supports of the minimum weight codewords of $\mathscr{C}^{\perp}$ form a $(t+1)$-design. In particular, taking $t=1$ recovers Corollary 1 of  Ly and Soljanin \cite{Ly}. 
\end{remark}


We next establish the tightness of the upper bound in Theorem \(\ref{thm:design}\) for systematic codes with explicitly characterized service rate regions (SRRs), specifically, first-order Reed–Muller codes and MDS codes.

\begin{enumerate}
\item[(a)] \textbf{First order Reed-Muller codes:} Let $m\geq 3$ be an integer. The dual of the first order binary Reed--Muller code $\mathrm{RM}(1,m)$ is $\mathrm{RM}(m-2,m)$, which is also the extended Hamming code with parameters $[2^m, 2^m-1-m, 4]$ (see Section 1.10 of \cite{Huffman}). By Theorem  15 of \cite[Chap. 2]{Sloane}, we note that the minimum Hamming weight codewords in $\mathrm{RM}(m-2,m)$ hold a $3$-$(2^m,4,1)$ design. Then, corresponding to a systematic generator matrix of $\mathrm{RM}(1,m)$, the sum of achievable service rates for any two  distinct data objects  $i,j\in[m+1]$, we get
$$\bigl(\lambda_i+\lambda_j\bigr)^{*}\;=\;\frac{2^{m}+4}{3}.$$
 When $m=3,$ by Theorem \ref{thm:design}, we get $\bigl(\lambda_i+\lambda_j\bigr)^{*}=4$  which is equivalent to the bound presented  in Theorem 3 of 
 Kazemi {\etal} \cite{Geo}, 
 
Moreover, for $m\geq 4$, by Theorem 3 of Kazemi {\etal} \cite{Geo},  the outer region containing SRR of $\mathrm{RM}(1,m)$ corresponding to a systematic generator matrix   is given by
\begin{equation}\label{reed}
    \sum_{i\in[m+1]\setminus I}\lambda_i
+\sum_{j\in I}(3\lambda_j-2)
\le 2^{m},
\qquad \text{ for all } I\subseteq[m+1].
\end{equation}
Consider $I=\{i,j\}$ and set $\lambda_\ell=0$ for all
$\ell\in [m+1]\setminus I$. By equation \eqref{reed}, we have
$$(3\lambda_i-2)+(3\lambda_j-2)\le 2^m, $$
that is, $$ \lambda_i+\lambda_j
\le
\frac{2^{m}+4}{3}. $$
Thus, when $|I|=2,$ the corresponding constraints of the SRR, characterized in \cite[Thm. 3]{Geo},  yield exactly the same inequality as the bound established in  Theorem \ref{thm:design}, and this bound is attainable. Equivalently, for every $I\subseteq[k]$ with $|I|=2,$ the constraint
$\sum_{i\in[m+1]\setminus I}\lambda_i+\sum_{j\in I}(3\lambda_j-2)\leq 2^{m+1}-1$
given in \cite[Thm. 3]{Geo} is tight.
\vspace{2mm}
\item[(b)] \textbf{MDS codes:}  Let $\mathscr{C}$ be an $[n,k,n-k+1]$ MDS code. Since $\mathscr{C}^\perp$ is also an MDS code, any $(k+1)$ elements subset of $[n]$ gives the support of a minimum weight codeword of $\mathscr{C}^\perp$ (see  Theorem $4$ of \cite[Chap. 11]{Sloane}). Then the supports of minimum weight codewords of $\mathscr{C}^\perp$ form the complete design $(k+1)$-$(n,\,k+1,\,1)$. By Lemma \ref{lem2.1}, we see that the supports of minimum weight codewords of $\mathscr{C}^\perp$ also form  $t$-design for any $t,$ where $1\le t\le k+1$.  Then for any $I\subseteq[k]$
with $|I|=t$ and 
  $t\ \le n- k$, we have 
\begin{equation*}
      \Bigl(\sum_{i\in I}\lambda_{i}\Bigr)^{*}= t+\frac{n-t}{d^{\perp}-1} =\;t+\frac{n-t}{k}.
\end{equation*}
Ly and Soljanin \cite{LyH} characterized the SRR of MDS codes by determining the complete service polytope associated with a generator matrix $G_\ell(n,k),$  where $\ell$ denotes the number of systematic servers under the condition $n\geq k+s.$ In \cite[Thm. 4]{LyH}, Ly and Soljanin consider the set $\mathcal{A}=\{j \in[\ell]~:~\lambda_j \geq 1\}$ and $\mathcal{B}=[\ell]\setminus \mathcal{A}.$  By taking $\ell=k,$ $\mathcal{A}=I(\subseteq [k])$   and $\lambda_j=0$ for all $j\in [k]\setminus I$  in \cite[Thm. 4]{LyH}, their result recovers precisely the bound obtained in Theorem \ref{thm:design} for the aggregate service rate of the data objects indexed by $I.$

In particular, for a systematic generator matrix $G_k(n,k),$ their characterization applies when $n\geq 2k.$ Our result, however, requires only $n\geq k+t,$ where $t=|I|.$ Thus, whenever $t<k,$ this condition is weaker than the requirement $n\geq 2k$ for the complete characterization of the systematic case in \cite{LyH}. Consequently, Theorem \ref{thm:design} also applies in the parameter range $k+t\leq n\leq 2k-2,$
for which the complete SRR is not known in general, as noted in \cite[Subsection E]{LyH}.  Nevertheless, in this range, Theorem 
\ref{thm:design}  determines the exact value of
$\left(\sum_{i\in I}\lambda_i\right)^*$
for every subset $I\subseteq[k]$ with  $|I|=t.$

The smallest parameter set illustrating this distinction is $(n,k,t)=(6,4,2).$ In this case, Theorem \ref{thm:design} gives the exact maximum aggregate service rate as $3,$ with all servers saturated at the optimum. In contrast, applying the vertex cover bound of \cite[Lemma 3]{LyH} yields the upper bound $4,$ which does not determine the exact optimum.


\end{enumerate}

\section{Applications: Codes Attaining the Upper Bound}\label{app}
In this section, we present additional families of codes for which the upper bound established in Theorem \ref{thm:design} is attained. 
\begin{enumerate}

\item[(I)] \textbf{Reed-Muller code $\mathrm{RM}(r-1,m)$:} Let $m$ and $r$ be positive integers satisfying $m\ge 3$ and $2\le r\le m-1$. It is well known that $\mathrm{RM}(r-1,m)^{\perp}=\mathrm{RM}(m-r,m)$ \cite[Chap. 13]{Sloane}. Further, we note that the minimum weight codewords of
$\mathrm{RM}(m-r,m)$ hold a
$3$-design (see page 400 of \cite{Sloane}). By Theorem \ref{thm:design}, for a systematic
generator matrix of $\mathrm{RM}(r-1,m),$  the
sum of achievable service rates for any two distinct data objects 
$i$ and $j$ is given by 
$$\bigl(\lambda_i+\lambda_j\bigr)^{*}
  =2+\frac{2^m-2}{2^{r}-1}=\frac{2^m+2^{r+1}-4}{2^{r}-1}.$$

\item[(II)] The extended binary Golay code with parameters $[24,12,8]$ is a self-dual code, and the supports of its minimum weight codewords hold a $5$-$(24, 8, 1)$ design (see Chapter 20 of \cite{Sloane}). Then, corresponding to a systematic generator matrix of this code, the bound on the cumulative service rate of data objects indexed by $I\subseteq[12]$, with $|I|=t\leq4$, is \[\sum_{i\in I}\lambda_i\leq t+\frac{24-t}{7}=\frac{24+6t}{7}. \]

 For $1\leq t\leq 4,$ the resulting values of $\mathcal{M}_1^{t-1},$ $\mathcal{M}_2^{t-1},$ the assigned weight $\tfrac{1}{t\mathcal{M}_2^{t-1}}$ to each node in $[24]\setminus I,$ the service rate $\lambda_i$ for each $i\in I,$ and the resulting aggregate service rate are summarized in TABLE \ref{table2} (recall the notations from the proof of Theorem \ref{thm:design} that for each $i\in I\subseteq [12]$,  $\mathcal{M}_1^{t-1}$ denotes the number of blocks containing $i$ and no point of $I\setminus\{i\}$, and, for each $v\in [24]\setminus I,$  $\mathcal{M}_2^{t-1}$ denotes the number of such blocks that additionally contain $v.$ Assigning weight $1$ to systematic node $i$ and weight equal to  $1/(t\mathcal{M}_2^{t-1})$ to each  recovery set of $u_i$ that is disjoint from $I$  yields a feasible allocation). 
\begin{table}[h]\label{table2}
\centering
\begin{tabular}{ccccccc}
\hline
$t$ & $\mathcal{M}_1^{t-1}$ & $\mathcal{M}_2^{t-1}$ &  $\tfrac{1}{t\mathcal{M}_2^{t-1}}$
    & $\lambda_i$ & $\sum_{i\in I}\lambda_i$ & Total load \\
\hline
$1$ & $253$ & $77$ & $1/77$  & $30/7\approx4.2857$ & $30/7\approx4.2857$ & $24$ \\
$2$ & $176$ & $56$ & $1/112$ & $18/7\approx2.5714$ & $36/7\approx5.1429$ & $24$ \\
$3$ & $120$ & $40$ & $1/120$ & $2$                 & $6$                 & $24$ \\
$4$ & $80$  & $28$ & $1/112$ & $12/7\approx1.7143$ & $48/7\approx6.8571$ & $24$ \\
\hline
\end{tabular}
\caption{Allocation data for the extended binary Golay code.}\label{table2}
\end{table}

\item[(III)] The extended ternary Golay code with parameters $[12,6,6]$ is a self-dual code, and  the supports of its minimum weight codewords hold a $5$-$(12, 6, 1)$ design (see Chapter 20 of \cite{Sloane}). Then, corresponding to a systematic generator matrix of this code, the bound on the aggregate service rate of data objects indexed by $I\subseteq[6]$, with $|I|=t\leq4$, is $$\sum_{i\in I}\lambda_i\leq t+\frac{12-t}{5}=\frac{12+4t}{5}. $$

For $1\leq t\leq 4,$ the resulting values of $\mathcal{M}_1^{t-1},$ $\mathcal{M}_2^{t-1},$ the assigned weight $\tfrac{1}{t\mathcal{M}_2^{t-1}}$ to each node in $[12]\setminus I,$ the service rate $\lambda_i$ for each $i\in I,$ and the resulting aggregate service rate are summarized in TABLE \ref{table1}. 
\textcolor{red}{
\begin{table}[h]
\centering
\begin{tabular}{ccccccc}
\hline
$t$ & $\mathcal{M}_1^{t-1}$ & $\mathcal{M}_2^{t-1}$ &  $\tfrac{1}{t\mathcal{M}_2^{t-1}}$
    & $\lambda_i$ & $\sum_{i\in I}\lambda_i$ & Total load \\
\hline
$1$ & $66$ & $30$ & $1/30$ & $16/5=3.2$ & $16/5=3.2$ & $12$ \\
$2$ & $36$ & $18$ & $1/36$ & $2$        & $4$         & $12$ \\
$3$ & $18$ & $10$ & $1/30$ & $8/5=1.6$  & $24/5=4.8$  & $12$ \\
$4$ & $8$  & $5$  & $1/20$ & $7/5=1.4$  & $28/5=5.6$  & $12$ \\
\hline
\end{tabular}
\caption{Allocation data for the extended ternary Golay code.}\label{table1}
\end{table}}

\item[(IV)] Let $q = 2^m$ with $m\ge 4$. Ding and Tang \cite[Thm. 34]{Tang} constructed a class of narrow sense BCH codes $\mathcal{D}$  over $\mathbb{F}_q$ with parameters $[q+1,q-5,d]$ and designed distance $4,$ where  
$d=6$ if $m$ is odd and $d=5$ if $m$ is even. In Theorem 35, they showed that $\mathcal{D}^{\perp}$ has parameters $[q+1,6,q-5].$ Further, in Theorem 37, they showed that  when $m\geq 5$ is odd,  the minimum weight codewords in  $\mathcal{D}$ and $\mathcal{D}^\perp$  support
 $4$-designs.   Then, corresponding to a systematic generator matrix of $\mathcal{D}$, the bound on the aggregate service rate of data objects indexed by $I\subseteq[q-5]$, with $|I|=t\leq 3$, is $$\sum_{i\in I}\lambda_i\leq t+\frac{q+1-t}{q-6}=\frac{q(t+1)-7t+1}{q-6}. $$

Further, corresponding to a systematic generator matrix of $\mathcal{D}^{\perp}$, the bound on the aggregate service rate of data objects indexed by $I\subseteq[6]$, with $|I|=t\leq 3$, is $$\sum_{i\in I}\lambda_i\leq t+\frac{q+1-t}{5}=\frac{q+1+4t}{5}.$$

\item[(V)]   Table \ref{tab:3designs} presents several classes of linear codes for which the supports of the minimum-weight codewords of their dual codes form $3$-designs. In particular, we consider families of NMDS, AMDS, cyclic,  BCH, and Steiner codes satisfying this design property. Throughout TABLE \ref{tab:3designs}, $\mathscr{C}$ denotes an $[n,k,d]$ linear code over $\mathbb{F}_q,$ while the entries marked with $\dagger$ correspond to $[n,k,d]$ linear codes over $\mathbb{F}_{q^2}.$ For each of these code families, we determine the exact value of $(\lambda_i+\lambda_j)^*$ for any two distinct data objects $u_i$ and $u_j,$ assuming that $\mathscr{C}$ is represented by a systematic generator matrix.  

\begin{table}[h]
\centering
\footnotesize
\setlength{\tabcolsep}{4pt}
\renewcommand{\arraystretch}{1.25}
\begin{tabular}{@{}llccccl@{}}
 \hline\multirow{3}{*}{}
\textbf{Family} & \textbf{Ref.} & \textbf{Result} &  \textbf{Field}&
\textbf{Parameter of}& $d^\perp$&$\bigl(\lambda_i+\lambda_j\bigr)^{*}$ \\ & & &  \textbf{size}&$\mathscr{C}$& & \\
\hline

\multicolumn{7}{@{}l}{\textit{(a) NMDS and AMDS families ($m\geq 4$)}}\\[2pt]\hline
\noalign{\vskip 2pt}
\multirow{2}{*}{NMDS}
 &  \cite{Ding} & Thm. $21$ & $q=3^m$
 & 
 $[q+1,q-3,4]$ & $q-3$
 & $\frac{3q-9}{q-4}$\vspace{1mm} \\& & Thm. $21$ & $q=3^m$
 &  
  $[q+1,4,q-3]$ & $4$ &
  $\frac{q+5}{3}$\vspace{1mm}\\\hline
  \noalign{\vskip 2pt}
\multirow{2}{*}{AMDS }
 &  \cite{Xu}&Thm.~$5$ & $q=p^m$
 &$[q+1, q-3, 4]$ & $q-p^{\ell}$ & $2+\frac{q-1}{ q-p^{\ell}-1}$ \vspace{1mm} \\   & & & 
 &$h<m$ \& $\gcd(h,m)=\ell$&  & \vspace{1mm}\\\hline
 \noalign{\vskip 2pt}
\multirow{2}{*}{NMDS }
 &  \cite{Xu} & Thm.~$7$ & $q=2^m$
 &$[q, 4, q-4]$  & $4$
 & $\frac{q+5}{3}$ \vspace{1mm} \\ & &Thm.~$7$ & $q=2^m$
 &$[q, q-4, 4]$& $ q-4$  & $\frac{3q-12}{q-5}$ \vspace{1mm}\\\hline
 
   \noalign{\vskip 2pt}
 
\multirow{2}{*}{NMDS}
 &\cite{Heng} & Thm.~$31$ & $q=2^m$
 &$[q, 5, q-5]$  & $5$
 & $\frac{q+6}{4}$ \vspace{1mm} \\ & &Thm.~$31$ & $q=2^m$
 &$[q, q-5, 5]$& $q-5$  & $\frac{3q-14}{q-6}$  \vspace{1mm}\\\hline
  \noalign{\vskip 2pt}
\multirow{2}{*}{NMDS}
 &\cite{Heng} & Thm.~$38$ \& $42$ & $q=2^m$
 &$[q, 6, q-6]$  & $6$
 & $\frac{q+8}{5}$ \vspace{1mm}  \\ & &Thm.~$38$  \& $42$ & $q=2^m$
 &$[q, q-6, 6]$& $q-6$  & $\frac{3q-16}{q-7}$ \vspace{1mm}\\\hline
  \noalign{\vskip 2pt}

AMDS$\dagger$
 &  \cite{Cao} & Thm.~$4$&$q=3^m$
 &$[q+1, q-5, 6]$   & $q-6$
 & $\frac{3q-15}{q-7}$ \\\hline
  \noalign{\vskip 2pt}
\multirow{2}{*}{NMDS$\dagger$} 
 &  \cite{Cao} & Thm.~$5$&$q=3^m$
 & $[q+1, 7, q-6]$   & $7$
 & $\frac{q+11}{6}$  \vspace{1mm} \\ & &Thm.~$5$ & $q=3^m$
 &$[q+1, q-6,7]$& $q-6$  & $\frac{3q-15}{q-7}$ \vspace{1mm}\\\hline
   \noalign{\vskip 2pt}

\multicolumn{7}{@{}l}{\textit{(b) Cyclic, BCH and Steiner families ($m\geq 2$)}}\\ [2pt] \hline
\noalign{\vskip 2pt}
 \multirow{3}{*}{Narrow-sense BCH }
 & \cite{Tang}& Thm.~$40$& $q=2^{2m+1}$
 & $[q+1,q-5, 5]$ &$q-5$
 & $\frac{3q-13}{q-6}$  \vspace{1mm} \\  & &Thm.~$40$& $q=2^{2m+1}$&$[q+1,6,q-5]$&  $ 5$& $\frac{q+7}{4}$\vspace{1mm}\\ & & & $m\geq 4$ even & & &\\\hline
  \noalign{\vskip 2pt}
\multirow{2}{*}{Antiprimitive BCH/AMDS}
  & \cite{Liu} & Thm.~$23$& $q=\delta^m$
 &  $[q+1,2\delta-2,q-2\delta+3]$ & $\delta+1$
 & $\frac{2\delta+q-1}{\delta}$ \vspace{1mm} \\ &&&$\delta$ is prime power&&\\\hline

   \noalign{\vskip 2pt}
 \multirow{2}{*}{Antiprimitive cyclic}
 & \cite{Xiang}& Thm.~$3$& $q=7^m$
 & $[q+1,4,q-7]$ &$4$
 & $\frac{q+5}{3}$  \vspace{1mm} \\ &&Thm.~$3$& $q=7^m$&$[q+1,q-3,4]$& $q-7$& $\frac{3q-17}{q-8}$\vspace{1mm}\\\hline
  \noalign{\vskip 2pt}
 \multirow{3}{*}{Cyclic }
 & \cite{Wang} & Thm.~$13$& $q=p^m$
 & $[q+1, 4,q-p^l]$  & $4$
 & $\frac{q+5}{3}$ \vspace{1mm} \\ & &Thm.~$13$& $q=p^m$
 &$[q+1, q-3, 4]$ &$q-p^l$ & $2+\frac{q-1}{ q-p^l-1}$\vspace{1mm}\\ & & & $p$ is odd  &$h<m$ \& $\gcd(h,m)=l$ & &\\\hline
  \noalign{\vskip 2pt}
 \multirow{2}{*}{ Negacyclic$\dagger$}
 & \cite{Wang} & Thm.~$22$&  $q\equiv 1\pmod{4}$
 &$[q^2+1, 4,q^2-q]$  & $4$
 &  $\frac{q^2+5}{3}$ \vspace{1mm} \\ & &Thm.~$22$& $q\equiv 1\pmod{4}$
 &$[q^2+1, q^2-3, 4]$ &$q^2-q$ & $\frac{3q^2-2q-3}{ q^2-q-1}$\vspace{1mm}\\\hline
  \noalign{\vskip 2pt}
 \multirow{2}{*}{ Negacyclic$\dagger$}
 & \cite{Wang} & Thm.~$32$& $q\equiv 3\pmod{4}$
 & $[q^2+1, 4,q^2-q]$  & $4$
 & $\frac{q^2+5}{3}$ \vspace{1mm}  \\ & &Thm.~$32$& $q\equiv 3\pmod{4}$
 &$[q^2+1, q^2-3, 4]$ &$q^2-q$ & $\frac{3q^2-2q-3}{ q^2-q-1}$\vspace{1mm}\\\hline
  \noalign{\vskip 2pt}
 \multirow{2}{*}{ Constacyclic$\dagger$}
 & \cite{Hu} & Thm.~$2$ \& $3$& $q=p^m$
 &$[q^2+1, 4,q^2-q]$  & $4$
 & $\frac{q^2+5}{3}$ \vspace{1mm} \\& & Thm.~$2$ \& $3$& $q=p^m$
 &$[q^2+1, q^2-3, 4]$  & $q^2-q$ 
 &$\frac{3q^2-2q-3}{ q^2-q-1}$\vspace{1mm}\\

\hline
\end{tabular}
\caption{Exact Values of $(\lambda_i+\lambda_j)^*$ for two distinct data objects in classes of linear codes whose dual codes have minimum-weight codewords supporting a $3$-design. Throughout, $p$ is a prime number and $\mathscr{C}$ denotes an $[n,k,d]$ linear code over $\mathbb{F}_q,$ while the entries marked with $\dagger$ correspond to $[n,k,d]$ linear codes over $\mathbb{F}_{q^2}.$ } 

\label{tab:3designs}
\end{table}

\end{enumerate}

\section{Conclusion and future work}\label{conclusion}

In this paper, we investigated aggregate service rates in coded distributed storage systems based on linear codes over arbitrary finite fields using the supports of the minimum weight codewords of the dual codes and combinatorial design theory. We derived upper bounds on the aggregate service rate for an arbitrary subset of data objects for both systematic and non-systematic generator matrices. In the systematic case, the resulting bound takes a simple form depending only on the code length, the dual distance, and the cardinality of the selected subset. This provides a subset-wise perspective on the First Dual Distance Bound and extends previously known results for the service rate of an individual data object. We further investigated conditions under which the derived bound is attained. Applications of these results to MDS, Reed–Muller, extended Golay, and other families of linear codes show that designs supported by minimum-weight dual codewords can yield exact aggregate service rates for subsets of data objects. A limitation of this approach is that the scarcity of known code families whose dual codes have minimum-weight codewords supporting $t$-designs of higher strength, particularly for $t\geq 4.$ Several interesting problems in this direction remain open. One natural direction is to derive tighter bounds on the SRR by incorporating additional structural information about the code and its recovery sets beyond the dual distance. Another important problem is to explicitly characterize the SRR for broader classes of linear codes, particularly cyclic, constacyclic, quasi-cyclic, and locally recoverable codes.


\section*{Acknowledgments}
 The  first author gratefully acknowledges the financial support provided by the Anusandhan National Research Foundation (ANRF), Government of India, under the NPDF scheme through Grant No. PDF/2025/000550. The second  author is grateful to the University Grants Commission (UGC), Government of India, for the financial support received under the UGC NET-SRF scheme.

{}
\end{document}